\documentclass{article}

\usepackage{amsmath, amsthm, amssymb}
\usepackage{cite}
\usepackage{url}
\usepackage{fullpage}
\usepackage{cleveref}
\usepackage{appendix}
\usepackage[letterpaper,margin=1in]{geometry}
\usepackage{color}
\usepackage{algorithm}
\usepackage{algorithmic}

\newtheorem{theorem}{Theorem}[section]
\newtheorem{lemma}[theorem]{Lemma}

\newcommand{\FullOrShort}{full}

\ifthenelse{\equal{\FullOrShort}{full}}{
	
	  \newcommand{\fullOnly}[1]{#1}
	  \newcommand{\shortOnly}[1]{}

  }{
	  \newcommand{\fullOnly}[1]{}
	  \newcommand{\shortOnly}[1]{#1}
  }

\begin{document}

\date{}

\author{Mohsen Ghaffari\\ \texttt{ghaffari@mit.edu}\\ MIT \and Bernhard Haeupler\\ \texttt{haeupler@mit.edu}\\ MIT \and Majid Khabbazian\\ \texttt{m.khabbazian@uwinnipeg.ca}\\ University of Winnipeg}

\title{A Bound on the Throughput of Radio Networks
}

\maketitle
\begin{abstract}
We consider the well-studied \emph{radio network model}: a synchronous model with a graph $G=(V,E)$ with $|V|=n$ where in each round, each node either transmits a packet, with length $B=\Omega(\log n)$ bits, or listens. Each node receives a packet iff it is listening and exactly one of its neighbors is transmitting. We consider the problem of $k$-message broadcast, where $k$ messages, each with $\Theta(B)$ bits, are placed in an arbitrary nodes of the graph and the goal is to deliver all messages to all the nodes. We present a simple proof showing that there exist a radio network with radius $2$ where for any $k$, broadcasting $k$ messages requires at least $\Omega(k\log n)$ rounds. That is, in this network, regardless of the algorithm, the maximum achievable broadcast throughput is $O(\frac{1}{\log n})$.
\end{abstract}

\section{The Model and the Problem Statement}
We consider the well-studied \emph{radio network model}, first intoruced by Chlamtac and Kutten~\cite{CK}. This connections in this model are presented by a graph $G=(V,E)$ with $|V|=n$. Moreover, the model is synchronous, i.e., the executions proceed in lock-step rounds. In each round, each node either transmits a packet, with length $B=\Omega(\log n)$ bits, or listens. Each node receives a packet iff it is listening and exactly one of its neighbors is transmitting a packet.

We consider the $k$-message broadcast problem, where $k$ messages, each consisting of $\Theta(B)$ bits, are placed in an arbitrary subset of the nodes of the graph and the goal is to deliver all of these messages to all the nodes. We do not require that the packets that nodes transmit (over the channel) are one of these messages. In other words, we do not restrict the algorithm to be a routing algorithm, and in particular, it can be a network coding algorithm~\cite{KM03}, where each of the transmitted packets is a combination of some of the messages.  

\section{The Bound}
We show the following theorem.

\begin{theorem}\label{thm:main} There exist a radio network with radius $2$ where for any $k$, broadcasting $k$ messages requires at least $\Omega(k\log n)$ rounds.
\end{theorem}

The main approach of our proof is similar to that of the proof of the $\Omega(\log^2 n)$ lower bound of Alon, Bar-Noy, Linial and Peleg~\cite{ABLP} on the time to broadcast a single message in radio networks with radius $2$. We remark that \Cref{thm:main} also follows from the $\Omega(n \log n)$ gossip lower bound proof of Gasienec and Potapov~\cite{GP02}, which itself is achieved by a reduction to the $\Omega(\log^2 n)$ lower bound of \cite{ABLP}. Our proof is direct and considerably shorter and simpler than the proof of~\cite{ABLP}.

In order to prove \Cref{thm:main}, we first work with bipartite networks $G=(V, E)$ where $V=S\cup R$, $S \cap R =\emptyset$, and there is no edge between two nodes of $S$, or between two nodes of $R$. We call the nodes in $S$ \emph{senders} and the nodes in $R$ the \emph{receivers}. We show the following theorem about bipartite radio networks.
 
\begin{lemma}\label{lem:core}
There exists a bipartite network $\mathsf{H}$ with less than $n$ nodes such that in each round, regardless of which nodes transmit, at most $O(\frac{1}{\log n})$ of receivers receive a packet. 
\end{lemma}

\begin{proof}
We consider a distribution over a family of bipartite graphs $\mathcal{G}$ where we have $|S|=n'=\sqrt{n}$ senders and $|R|=\frac{n' \log n}{2}$ receivers. The receivers are divided into $\frac{\log n}{2}$ classes, each of equal size $n'$. For each $i \in \{1, 2, \dots, \frac{\log n}{2}\}$, the receivers of class $i$ have degree exactly $2^i$ in each graph of family $\mathcal{G}$. To present the distribution, we explain how to draw a random graph from this distribution. In a random graph $G \in \mathcal{G}$, the connections are chosen randomly as follows: for each $i \in \{1, 2, \dots, \frac{\log n}{2}\}$, each receiver in the $i^{th}$ receiver class is connected to $2^i$ randomly chosen senders. The choices of different receivers are independent. Note that the size of each graph in this family is $n'(1+\frac{\log n}{2}) <n$. 

Consider a random graph $G\in \mathcal{G}$. We claim that, with probability at least $1-e^{-2n'}$, $G$ has the property that in each round at most a $O(\frac{1}{\log n})$ fraction of the receiver nodes receive a packet, regardless of which set nodes transmit. 

To prove this claim, we first study the receptions in $G$ when a fixed subset $S'$ of senders transmit. More precisely, we calculate the expected number of receivers in $G$ that receive a packet if exactly senders in $S'$ transmit. 
Consider a receiver $v$ with degree $\Delta$. Receiver $v$ receives a packet iff exactly one of its sender neighbors is in set $S'$. If $\Delta>n'-|S'|+1$, then clearly $v$ does not receiver a packet. Suppose that $\Delta\leq n'-|S'|+1$. Then, the probability that $G$ is such that $v$ receives a packet is exactly
\begin{equation}
\label{equ:pDelta}
\begin{split}
  P_{\Delta}(|S'|)=\frac{\binom{|S'|}{1} \binom{n'-|S'|}{\Delta - 1}}{\binom{n'}{\Delta}}
  &= \frac{|S'|\Delta}{n'}\prod_{i=1}^{\Delta-1} (1 - \frac{|S'|-1}{n'-i}) \leq \frac{|S'|\Delta}{n'}(1 - \frac{|S'|-1}{n'-1})^{\Delta-1}\\
  &\leq \frac{|S'|\Delta}{n'}\exp\left(-\frac{|S'|-1}{n'-1}(\Delta-1)\right)\leq \frac{|S'|\Delta}{n'}\exp\left(-\frac{|S'|-1}{n'}(\Delta-1)\right)\\
  &= \frac{|S'|\Delta}{n'}\exp\left(-\frac{|S'|}{n'}\Delta\right)\exp\left(\frac{|S'|+\Delta-1}{n'}\right)\\
  &\leq e\cdot\frac{|S'|\Delta}{n'}\exp\left(-\frac{|S'|}{n'}\Delta\right).
\end{split}
\end{equation}
For each $i \in \{1, 2, \dots, \frac{n'\log n}{2}\}$ receiver, let $X_i$ be an indicator random variable which is each equal to $1$ iff the $i^{th}$ receiver receives a packet. Also define the random variable $X=\sum_{i=1}^{\frac{n'\log n}{2}} X_i$. Let $\Delta^*=2^{\lfloor\log(\frac{n'}{|S'|})\rfloor}\leq \frac{n'}{|S'|}$.
Using (\ref{equ:pDelta}), we have
\begin{equation}
\label{equ:max5}
\begin{split}
  \mathbb{E}[X] = n' \sum_{i=1}^{\log n'}P_{2^i}(|S'|)
  &\leq en'\cdot\sum_{i=1}^{\log n'}\frac{|S'|2^i}{n'}
     \exp\left(-\frac{|S'|}{n'}2^i\right)\\
  &= en'\cdot\left(\sum_{i=1}^{\log\Delta^*}\frac{|S'|2^i}{n'}
     \exp\left(-\frac{|S'|}{n'}2^i\right)+\sum_{i=\log\Delta^*+1}^{\log n'}\frac{|S'|2^i}{n'}
     \exp\left(-\frac{|S'|}{n'}2^i\right)\right)\\
  &\leq en'\cdot\left(\sum_{j=0}^{\infty}\frac{1}{2^j}
    +\sum_{j=0}^{\infty}\frac{2^{j+1}}{e^{2^{j}}}\right)<10n'.
\end{split}
\end{equation}
Note that the random variables $X_i$ are independent as the neighbors of different receivers are chosen independently. Thus, we can use a chernoff bound and infer that $Pr(X>20 n')<e^{-3n'}$.
That is, when exactly nodes in set $S'$ are transmitting, with probability at least $1-e^{-3n'}$, random graph $G$ is such that at most $20n'$ receivers receive a packet. 

Now note that the total number of choices for set $S'$ is $2^{n'}$. Therefore, by a union bound over all choices of set $S'$, we get that with probability at least $1-e^{-3n'}\cdot2^{n'}>1-e^{-2n'}$, the random graph $G$ is such that no set $S'$ can deliver a packet to more than $20n'$ receivers. This in particular shows that there exists a bipartite graph $\mathsf{H}$ in this family such that no set $S'$ can deliver a packet to more than $20n'$ receivers. Since there are $\frac{n' \log n}{2}$ receivers, we get that in $\mathsf{H}$, there does not exist a subset of senders which their transmission delivers a packet to more than an $\frac{40}{\log n}$ fraction of the receiver nodes.
\end{proof}

\begin{proof}[Proof of \Cref{thm:main}] Consider network $\mathsf{H}$ proven to exist in \Cref{lem:core} and let $\eta$ be the size of $\mathsf{H}$, i.e., $\eta = n'(1+\frac{\log n}{2})$. We construct network $H'$ based on network $\mathsf{H}$ as follows: we add one source node $s$ and $n-\eta-1$ \emph{void} nodes and we connect node $s$ to all senders and all void nodes. Clearly $H'$ has radius $2$. Put $k$ messages in the source node $s$. For each receiver node $u$, in order for $u$ to have all the $k$ messages, it must receive at least $\Omega(k)$ packets. Note that this holds for any algorithm including network coding algorithms. Since receiver nodes are only connected to the sender nodes, from \Cref{lem:core}, we get that in each round at most $O(\frac{1}{\log n})$ of receivers receive a packet (any packet). Thus, it takes at least $\Omega(k\log n)$ rounds till each receiver has all the $k$ messages.    
\end{proof}

\bibliographystyle{acm}
\bibliography{Bdata}

\end{document}